\documentclass[submission]{eptcs}

\newif\ifappendix\appendixtrue

\appendixfalse

\usepackage{stmaryrd}
\usepackage{amsfonts}
\usepackage{amsmath}
\usepackage{amsthm}

\ifappendix
\usepackage[breakable]{tcolorbox}
\newenvironment{proof*}{%
	\begin{tcolorbox}[colback=green!12,colframe=white,/tcb/breakable]%
	\textit{Proof}. \\%
}{%
	\qed%
	\end{tcolorbox}%
}

\fi

\newcommand{\lrangle}[1]{\langle #1 \rangle}

\newcommand{\brtext}[1]{[\textrm{\small #1}]}

\newcommand{\trule}[1]{{\footnotesize\brtext{#1}}}

\newcommand{\noi}{\noindent}

\newcommand{\bnfis}{\;\;::=\;\;}
\newcommand{\bnfbar}{\;\;\;|\;\;\;}

\newcommand{\tree}[2]{
\ensuremath{\displaystyle
		\frac
		{
			#1
		}{
			\raisebox{-0.4mm}{$\displaystyle{#2}$}
		}
	}
}

\newcommand{\vect}[1]{\tilde{#1}}

\newcommand{\set}[1]{\{#1\}}
\newcommand{\es}{\emptyset}

\newcommand{\setbar}{\ \ |\ \ }

\newcommand{\dom}[1]{\mathtt{dom}(#1)}

\newcommand{\freev}[1]{\lrangle{#1}}
\newcommand{\boundv}[1]{(#1)}

\newcommand{\send}[1]{\overline{#1}}
\newcommand{\ol}[1]{\overline{#1}}
\newcommand{\receive}[1]{#1}

\newcommand{\pisend}[2]{\send{#1}#2.}
\newcommand{\pirecv}[2]{\receive{#1}#2.}

\newcommand{\inact}{\mathbf{0}}

\newcommand{\Par}{\;|\;}
\newcommand{\news}[1]{(\nu\ #1)}
\newcommand{\newsp}[2]{(\nu\ #1)(#2)}
\newcommand{\varp}[1]{#1}
\newcommand{\rec}[2]{\mu \varp{#1}. #2}

\newcommand{\fn}[1]{\mathtt{fn}(#1)}

\newcommand{\scong}{\equiv}

\newcommand{\red}{\longrightarrow}

\newcommand{\subst}[2]{\set{#1/#2 }}

\newcommand{\context}[2]{#1[#2]}

\newcommand{\vart}[1]{\mathsf{#1}}

\newcommand{\ssep}{;}

\newcommand{\outses}{!}
\newcommand{\inpses}{?}
\newcommand{\selses}{\oplus}
\newcommand{\brases}{\&}
\newcommand{\dual}[1]{\overline{#1}}
\newcommand{\cat}{\cdot}

\newcommand{\bout}[2]{#1 \outses \freev{#2} \ssep}
\newcommand{\binp}[2]{#1 \inpses \boundv{#2} \ssep}
\newcommand{\bsel}[2]{#1 \selses #2 \ssep}
\newcommand{\bbra}[2]{#1 \brases \set{#2}}

\newcommand{\tfont}[1]{\mathtt{#1}}

\newcommand{\chtype}[1]{\lrangle{#1}}

\newcommand{\trec}[2]{\mu\vart{#1}.#2}

\newcommand{\tinact}{\tfont{end}}

\newcommand{\Ga}{\Gamma}

\newcommand{\proves}{\vdash}

\newcommand{\by}[1]{\stackrel{#1}{\longrightarrow}}

\newcommand{\true}{\sessionfont{tt}}

\newcommand{\s}{\ensuremath{s}}

\newif\ifny\nyfalse

\newif\ifdm\dmtrue

\newif\ifrhu\rhutrue

\newcommand{\chequiv}{\stackrel{.}{\leftrightarrow}}
\newcommand{\outcon}{\stackrel{.}{\prec}}
\newcommand{\incon}{\stackrel{.}{\succ}}
\newcommand{\freshin}{\#}
\newcommand{\psiout}[2]{\ol{#1} #2.}
\newcommand{\psiinp}[3]{#1(\lambda #2) #3.}

\newcommand{\case}{\mathbf{case}}
\newcommand{\casesep}{\mathrel{[\hspace{-0.1ex}]}}

\newcommand{\terms}{\mathbf{T}}
\newcommand{\conditions}{\mathbf{C}}
\newcommand{\assertions}{\mathbf{A}}

\newcommand{\encode}[1]{\llbracket #1\rrbracket}
\newcommand{\psiunit}{\mathbf{1}}
\newcommand{\frames}{\triangleright}

\newcommand{\psicong}{\scong_\Psi}

\newcommand{\pass}[1]{\llparenthesis #1 \rrparenthesis}

\newcommand{\spos}[1]{#1^+}
\newcommand{\sneg}[1]{#1^-}

\newcommand{\seq}[1]{\widetilde{#1}}
\newcommand{\ctxhole}{[\hspace{0.8pt}]}

\renewcommand{\by}[1]{\xrightarrow{#1}}
\renewcommand{\true}{\mathsf{true}}

\newcommand{\contractverticalspace}{\vspace{-0.5em}}

\newtheorem{theorem}{Theorem}[section]

\newtheorem{lemma}{Lemma}[section]
\newtheorem{definition}{Definition}[section]

	\title{Session Types for Broadcasting}
	\author{
		Dimitrios Kouzapas
		\institute{University of Glasgow}
		\email{dimitrios.kouzapas@glasgow.ac.uk}
	\and
		Ram\={u}nas Gutkovas
		\institute{Uppsala University}
		\email{ramunas.gutkovas@it.uu.se}
	\and
		Simon J. Gay
		\institute{University of Glasgow}
		\email{simon.gay@glasgow.ac.uk}
	}

\def\titlerunning{Session Types for Broadcasting}
\def\authorrunning{D. Kouzapas, R. Gutkovas \& S. J. Gay}

\begin{document}
	\maketitle

	\begin{abstract}
	Up to now session types have been used under the
	assumptions of point to point communication, to ensure the
	linearity of session endpoints, and reliable communication,
	to ensure send/receive duality.
	In this paper we define  a session type theory for
	broadcast communication semantics that by definition
	do not assume point to point and reliable communication.
	Our session framework
	lies on top of the parametric framework of broadcasting
	$\psi$-calculi, giving insights on developing session
	types within a parametric framework. Our session type theory
	enjoys the properties of soundness and safety.
	We further believe that the solutions proposed
	will eventually provide a deeper understanding
	of how session types principles should be applied
	in the general case of communication semantics.
\end{abstract}

	\section{Introduction}


Session types 
\cite{honda.vasconcelos.kubo:language-primitives,
  yoshida.vasconcelos:language-primitives, HYC08} allow communication
protocols to be specified as types and verified by type-checking. Up
to now, session type systems have assumed reliable, point to point 
message passing communication. Reliability
is important to maintain send/receive duality, and
point to point communication
is required to ensure session endpoint linearity.

In this paper we propose a session type system
for unreliable broadcast communication.  Developing such a system was
challenging for two reasons: (i) we needed to extend binary session
types to handle unreliability as well as extending the notion of
session endpoint linearity, and (ii) the reactive control flow of a
broadcasting system drove us to consider typing patterns of
communication interaction rather than communication prefixes. The key
ideas are (i) to break the symmetry between the $\spos{s}$ and
$\sneg{s}$ endpoints of channel $s$, allowing $\spos{s}$ (uniquely
owned) to broadcast
and gather, and $\sneg{s}$ to be shared; (ii) to implement (and type) the gather operation as an
iterated receive. We retain the standard binary session type constructors.

We use $\psi$-calculi
\cite{DBLP:conf/lics/BengtsonJPV09} as
the underlying process framework, and specifically we use the
extension of the $\psi$-calculi family with broadcast
semantics \cite{conf/sefm/BorgstromHJRVPP11}.
$\psi$-calculi provide a parametric process calculus
framework for extending the semantics of the $\pi$-calculus with
arbitrary data structures and logical assertions. Expressing our work
in the $\psi$-calculi framework allows us to avoid defining a new
operational semantics, instead defining the semantics of our broadcast
session calculus by translation into a broadcast $\psi$-calculus.  
Establishing a link between session types and $\psi$-calculi is
therefore another
contribution of our work.

{\bf Intuition through Demonstration.}
We demonstrate the overall intuition by means of an example.
For the purpose of the demonstration we imply a set of semantics,
which we believe are self explanatory. Assume types $S =
!T;?T;\tinact$, $\dual{S} = ?T;!T;\tinact$ for some data type $T$, and
typings $\spos{s}:S$, $\sneg{s}:\dual{S}$, $a:\chtype{S}$, $v:T$. The
session type prefix $!T$ means \emph{broadcast} when used by
$\spos{s}$, and \emph{single destination send} when used by
$\sneg{s}$. Dually, $?T$ means \emph{gather} when used by $\spos{s}$,
and \emph{single origin receive} when used by $\sneg{s}$.

{\bf Session Initiation} through broadcast,
creating an arbitrary number of receiving endpoints:

\begin{tabular}{c}
	$\pisend{a}{\sneg{s}} P_0 \Par \pirecv{a}{x} P_1 \Par \pirecv{a}{x} P_2 \Par \pirecv{a}{x} P_3 \red P_0 \Par P_1 \subst{\sneg{s}}{x} \Par P_2 \subst{\sneg{s}}{x} \Par \pirecv{a}{x} P_3$
\end{tabular}

\noi Due to unreliability, $\pirecv{a}{x} P_3$ did not initiate the session. We denote the initiating
and accepting session endpoint as $\spos{s}$ and $\sneg{s}$ respectively.

{\bf Session Broadcast} from the
$\spos{s}$ endpoint results in multiple $\sneg{s}$ endpoints 
receiving:

\begin{tabular}{c}
	$\bout{\spos{s}}{v} P_0 \Par \binp{\sneg{s}}{x} P_1 \Par \binp{\sneg{s}}{x} P_2 \Par \binp{\sneg{s}}{x} P_3 \red P_0 \Par P_1 \subst{v}{x} \Par P_2 \subst{v}{x} \Par \binp{\sneg{s}}{x} P_3$
\end{tabular}

\noi Due to unreliability, a process (in the above reduction, process
$\binp{\sneg{s}}{x} P_3$) might not receive a message. In this case
the session endpoint that belongs to process $\binp{\sneg{s}}{x} P_3$
is considered broken, and later we will introduce a recovery mechanism.

{\bf Gather:} The next challenge is to achieve the sending of values
from the $\sneg{s}$ endpoints to the $\spos{s}$ endpoint. The gather
prefix $\binp{\spos{s}}{x} P_0$ is translated (in
Section~\ref{sec:encoding}) into a process that iteratively
receives messages from the $\sneg{s}$ endpoints, non-deterministically
stopping at some point and passing control to $P_0$.

{\small
\begin{tabular}{c}
	$\binp{\spos{s}}{x} P_0 \Par \bout{\sneg{s}}{v_1} P_1 \Par \bout{\sneg{s}}{v_2} P_2 \Par \bout{\sneg{s}}{v_3} P_3 \red^*  P_0' \Par P_1 \Par \bout{\sneg{s}}{v_2} P_2 \Par P_3$
\end{tabular}
}\\
\noi with $P_0 \subst{\set{v_1,v_3}}{x} \red P_0'$.


\noi After two reductions the messages from processes
$\bout{\sneg{s}}{v_1} P_1$ and $\bout{\sneg{s}}{v_3} P_3$
had been received by the $\spos{s}$ endpoint.
On the third reduction the $\spos{s}$
endpoint decided not to wait for more messages and
proceeded with its session non-deterministically,
resulting in a broken sending endpoint ($\bout{\sneg{s}}{v_2} P_2$),
which is predicted by the unreliability of the 
broadcast semantics. The received messages, $v_1$ and $v_2$, were
delivered to $P_0$ as a set.

{\bf Prefix Enumeration:}
The above semantics, although capturing broadcast session
initiation and interaction, still violate
session type principles due to the unreliability of
communication:

{
\small
\begin{tabular}{rcl}
	$\bout{\spos{s}}{v_1} \bout{\spos{s}}{v_2} \inact \Par \binp{\sneg{s}}{x} \binp{\sneg{s}}{y} \inact \Par \binp{\sneg{s}}{x} \binp{\sneg{s}}{y} \inact$ &$\red$&\\
	$\bout{\spos{s}}{v_2} \inact \Par \binp{\sneg{s}}{y} \inact \Par \binp{\sneg{s}}{x} \binp{\sneg{s}}{y} \inact$ &$\red$&
	$\inact \Par  \inact \Par \binp{\sneg{s}}{y} \inact$
\end{tabular}
}

\noi The first reduction produced a broken endpoint,
$\binp{\sneg{s}}{x} \binp{\sneg{s}}{y} \inact$,
while the second reduction reduces the broken endpoint.
This situation is
not predicted by session type principles. To solve this problem
we introduce an enumeration on session prefixes:

\begin{tabular}{l}
	\small
	$\bout{(\spos{s}, 1)}{v_1} \bout{(\spos{s}, 2)}{v_2} \inact \Par \binp{(\sneg{s}, 1)}{x} \binp{(\sneg{s}, 2)}{y} \inact \Par \binp{(\sneg{s}, 1)}{x} \binp{(\sneg{s}, 2)}{y} \inact \red$\\
	\small
	$\bout{(\spos{s}, 2)}{v_2} \inact \Par \binp{(\sneg{s}, 2)}{y} \inact \Par \binp{(\sneg{s}, 1)}{x} \binp{(\sneg{s}, 2)}{y} \inact \red
	\inact \Par  \inact \Par \binp{(\sneg{s}, 1)}{x} \binp{(\sneg{s}, 2)}{y} \inact$
\end{tabular}

\noi The intuitive semantics described in this example
are encoded in the $\psi$-calculi framework. From this
it follows that
all the operational semantics, typing system and
theorems are stated using the $\psi$-calculus 
framework.

{\bf Contributions.}
This paper is the first to propose session types as a
type meta-theory for the $\psi$-calculi. Applying session
semantics in such a framework meets the ambition that
session types can effectively describe general communication
semantics. A step further is the development of a session
type framework for broadcast communication semantics. It 
is the first time that session types escape the assumptions
of point to point communication and communication reliability.
We also consider as a contribution the fact that we use
enumerated session prefixes in order to maintain consistency
of the session communication. We believe that
this technique will be applied in future session type 
systems that deal with unreliable and/or unpredictable
communication semantics.

{\bf Related Work.}  Carbone \emph{et al.} \cite{CarboneHY08} extended binary
session types with exceptions, allowing both parties in a session to
collaboratively handle a deviation from the standard
protocol. Capecchi \emph{et al.}  \cite{CGY2014} generalized a similar
approach to multi-party sessions. In contrast, our recovery processes
allow a broadcast sender or receiver to autonomously handle a failure
of communication. Although it might be possible to represent
broadcasting in multi-party session type systems, by explicitly
specifying separate messages from a single source to a number of
receivers, all such systems assume reliable communication for every message.

		\section{Broadcast Session Calculus}
	\label{syntax}


	We define an intuitive syntax for our calculus. The
	syntax below will be encoded in the $\psi$-calculi framework so
	that it will inherit the operational semantics.

{\small
	\begin{tabular}{rcllcll}


		$P, R$	& $\bnfis$  & $\pisend{a}{\sneg{s}} P \bnfbar \pirecv{a}{x} P \bnfbar \bout{\spos{s}}{v} P
			\bnfbar \bout{\sneg{s}}{v} P \bnfbar \binp{\spos{s}}{x} P \bnfbar \binp{\sneg{s}}{x} P$\\

			& $\bnfbar$ & $\bsel{\spos{s}}{l} P \bnfbar \bbra{\sneg{s}}{l_i:P_i}  \bnfbar P \bowtie R \bnfbar \inact
			\bnfbar \rec{X}{P} \bnfbar \varp{X} \bnfbar P \Par P \bnfbar \news{n} P$

	\end{tabular}
}

	Processes $\pisend{a}{\sneg{s}} P$, $\pirecv{a}{x} P$ are
        prefixed with session initiation operators that interact
        following the broadcast semantics. Processes
        $\bout{\spos{s}}{v} P$, $\bout{\sneg{s}}{v} P$
        define two different sending patterns. For the $\spos{s}$
        endpoint we have a broadcast send. For the $\sneg{s}$ endpoint
        we have a unicast send. Processes
        $\binp{\spos{s}}{x} P$, $\binp{\sneg{s}}{x} P$
        assume gather (i.e.\ the converse of broadcast send) and
        unicast receive, respectively. 
	We allow selection and branching
        $\bsel{\spos{s}}{l} P$, $\bbra{\sneg{s}}{l_i:P_i}$
        only for broadcast semantics from the $\spos{s}$ to the
        $\sneg{s}$ endpoint. Each process can carry a recovery
	process $R$ with the operator $P \bowtie R$. The process
        can proceed non-deterministically to recovery if the session
        endpoint is broken due to the unreliability of the
        communication. Process $R$ is carried along as process
        $P$ reduces its prefixes.
	The rest of the processes are standard
        $\pi$-calculus processes.
	
	Structural congruence is defined over the abelian monoid
        defined by the parallel operator $(\Par)$ and the inactive
        process $(\inact)$ and additionally satisfies the rules:

	\begin{tabular}{c}
		$\news{n} \inact \scong \inact \qquad P \Par \news{n}Q \scong \newsp{n}{P \Par Q}$ if $n \notin \fn{P}$
	\end{tabular}

	\section{Broadcast $\psi$-Calculi}

Here we define the parametric framework of $\psi$-calculi for broadcast.
For a detailed description of $\psi$-calculi we refer the reader to \cite{DBLP:conf/lics/BengtsonJPV09}. 

We fix a countably infinite set of names $\mathcal{N}$ ranged over by $a,b,x$.
$\psi$-calculi are parameterised over three nominal sets: terms ($\terms$ ranged over by $M,N,L$),
conditions ($\conditions$ ranged over by $\varphi$), and assertions ($\assertions$ ranged over by $\Psi$);
and operators: channel equivalence, broadcast output and input connectivity
$\chequiv,\outcon, \incon : \terms \times \terms \rightarrow \conditions$,
assertion composition $\otimes: \assertions\times\assertions\rightarrow \assertions$,
unit $\psiunit \in \assertions$, entailment relation $\proves \assertions\times\conditions$,
and a substitution function substituting terms for names for each set.
The channel equivalence is required to be symmetric and transitive, and
assertion composition forms abelian monoid with $\psiunit$ as the unit element. 
We do not require output and input connectivity be symmetric, i.e.,
$\Psi \proves M \outcon N$ is not equivalent to $\Psi \proves N \incon M$, however
for technical reasons require that the names of $L$ should be included in $N$ and $M$
whenever $\Psi\proves N \outcon L$ or $\Psi\proves L \incon M$.
The agents are defined as follows
\contractverticalspace
{
\small
\[
	P, Q	\bnfis \psiinp{M}{\seq{a}}{N}P
			\bnfbar \psiout{M}{N}P
			\bnfbar \case\,\varphi_1 : P_1 \casesep \dots \casesep \varphi_n : P_n
			\bnfbar \pass{\Psi} 
			\bnfbar (\nu a)P
			\bnfbar P \Par Q
			\bnfbar !P
\]
}%
where $\seq{a}$ bind into $N$ and $P$. The assertions in the case and replicated agents are required to be guarded.
We abbreviate the case agent as $\case\,\seq{\varphi}:\seq{P}$;
we write 
$\inact$ for $\pass{\psiunit}$, we also write $a \freshin X$ to intuitively
mean that name $a$ does not occur freely in $X$.


We give a brief intuition behind the communication parameters:  Agents
unicast
whenever their subject of their prefixes are channel equivalent, to give an
example, $\psiout{M}{L}P$ and $\psiinp{N}{\seq{a}}{K}Q$ communicate whenever
$\Psi\proves M\chequiv N$.  In contrast, broadcast communication is mediated by
a broadcast channel, for example, the agents
$\psiout{M}{N}P$ and 
$\psiinp{M_i}{\seq{a_i}}{N_i}P_i$ (for $i > 0$)
communicate if they can broadcast and receive from the same
channel $\Psi \proves M \outcon K$ and $\Psi\proves K \incon M_i$.

In addition to the standard structural congruence laws of pi-calculus we
define the following, with the assumption that $a \mathrel{\freshin}
\seq{\varphi}, M, N, \seq{x}$ and $\pi$ is permutation of a sequence.
\contractverticalspace
\[
	\begin{array}{cccc}
	(\nu a)\case\,\seq{\varphi}: \seq{P} \psicong \case\,\seq{\varphi} : \seq{(\nu a)P} &
	\case\,\seq{\varphi}: \seq{P} \psicong \case\,\pi\cdot(\seq{\varphi} : \seq{P}) \\
	\psiout{M}{N}(\nu a)P \psicong (\nu a)\psiout{M}{N}P  &
	\psiinp{M}{\seq{x}}{N}(\nu a)P \psicong (\nu a)\psiinp{M}{\seq{x}}{N}P \\
	\end{array}
\]
\contractverticalspace

The following is a reduction context with two types of numbered holes (condition
hole $\hat\ctxhole$ and process hole $\ctxhole$) such that no two
holes of the same type have the same number.
\contractverticalspace
\[
\textstyle
C \bnfis (\case\,\hat{\ctxhole}_j : C \casesep \seq{\varphi} : \seq{P}) \Par C 
	\;\bnfbar\; \prod_{k > 0}\ctxhole_{i_k}
\]

The filling of the holes is defined in the following way:
filling a process (resp. condition) hole with a assertion guarded process
(resp. condition) taken from the number position of a given sequence.
We denote filling of holes as $C[(\varphi_i)_{i\in I} ; (P_j)_{j\in J}
  ; (Q_k)_{k\in K}]$ where the first component is for filling the
condition holes and the other two are for filling process holes. 

We require that $I$ is equal to the numbering set of condition holes and
furthermore $J$ and $K$ are disjoint and their union is equal to the numbering
set of context for the process holes.
We also require that every $J$ numbered hole is either in parallel
with any of the $K$ holes or is parallel to $\case$ where recursively
a $K$ numbered hole can be found. When the numbering is understood we simply write
$C[\seq{\varphi}; \seq{P}; \seq{Q}]$.

In the following we define reduction semantics of $\psi$-calculi, in addition to the standard labelled transition semantics \cite{DBLP:conf/lics/BengtsonJPV09}.  The two rules
describe unicast and broadcast semantics. We identify agents up to
structural congruence, that is, we also assume the rule such that two agents
reduce if their congruent versions reduce.  In the broadcast rule, if for some
$a \in \seq{a}$, $a \in n(K)$, then $\seq{b} = \seq{a}$, otherwise $\seq{b} =
\seq{a} \setminus n(N)$.
To simplify the presentation we abbreviate $\prod\seq{\pass\Psi}$ as
$\hat{\pass\Psi}$ and $\otimes_i \Psi_i$ as $\hat\Psi$. We prove that reductions correspond to silent and broadcast transitions.

\noindent
\begin{center}
\begin{tabular}{c}
\small
	$
		\tree{
			\textstyle
			 N' = N[\seq{x} := \seq{L}] \text{ and } \hat\Psi\proves M\chequiv M' \text{ and } 
			\forall i .\hat\Psi\proves\varphi_i
		}{
			\textstyle
			(\nu\seq{a})(C[\seq{\varphi};\; \seq{R};\; \psiinp{M}{\seq{x}}{N}P ,\, \psiout{M'}{N'}Q] \Par \hat{\pass{\Psi}})
				\quad\rightarrow\quad
				(\nu\seq{a})(P[\seq{x} := \seq{L}] \Par Q \Par \prod \seq{R} \Par \hat{\pass{\Psi}})
		}
	$
	\\[5mm]
	\small
	$
		\tree{
			\textstyle
			\hat\Psi \proves M\outcon K
			\text{ and }
			\forall i. \hat\Psi\proves K \incon M'_i
			\text{ and } N'_i[\seq{x}_i := \seq{L}_i] = N
			\text{ and } \forall j. \hat\Psi\proves\varphi_j
		}{
			\textstyle
			(\nu\seq{a})(C[\seq{\varphi} ;\; \seq{R} ;\; \psiout{M}{N}P, (\seq{\psiinp{M'}{\seq{x}}{N'}Q})] \Par \hat{\pass\Psi})
				 \quad\rightarrow\quad
				(\nu\seq{b})(P \Par \prod_{i} Q_i[\seq{x}_i := \seq{L}_i] \Par \prod \seq{R}\Par \hat{\pass\Psi})
		}
	$
\end{tabular}
\end{center}

\begin{theorem}
    \label{thm:lts-red-correspondence}
	Let $\alpha$ be either a silent or broadcast output action.
    Then, $\psiunit \frames P \by\alpha P' \text{ iff } P \rightarrow P'$
\end{theorem}
\begin{proof}[Proof Sketch]
The complicated direction is $\implies$. One needs to prove similar results for the
other actions, and then demonstrate that they in parallel have the right form.
\end{proof}


	\section{Translation of Broadcast Calculus to Broadcast $\psi$-Calculus}
\label{sec:encoding}

The semantics for the broadcast session calculus are given as an instance of the
$\psi$-calculi with broadcast \cite{conf/sefm/BorgstromHJRVPP11}.
To achieve this effect we define a translation between the syntax of \S~\ref{syntax}
and a particular instance of the $\psi$-calculi. Operational semantics are then
inherited by the $\psi$-calculi framework.


%
%
We fix the set of labels $\mathcal{L}$ and ranged over by $l,l_1,l_2\dots$. The following
are the nominal sets
{\small
\[
\begin{array}{rcl}
	\terms &=& \mathcal{N} \cup \set{*} \cup \set{(n^p, k), (n^p, i), (n^p, k, \mathsf{u}), (n^p, l, k), n\cdot k \setbar n, k \in \terms \wedge i \in \mathbb{N} \wedge l \in \mathcal{L} \wedge p \in \set{+, -}}\\
	\conditions &=& \set{t_1 \chequiv t_2, t_1 \outcon t_2, t_1 \incon t_2 \setbar t_1, t_2 \in \terms} \cup \set{\true} \\
	\assertions &=& \terms \rightarrow \mathbb{N}
\end{array}
\]
}
We define the $\otimes$ operator (here defined as multiset union) and the $\proves$ relation:
\[
\small
\begin{array}{c|c}
        (f \otimes g)(n) = \left\{\begin{array}{ll}
                        f(n) + g(n) & \text{if $n \in dom(f) \cap dom(g)$} \\
                        f(n) & \text{if $n \in dom(f)$}\\
                        g(n) & \text{if $n \in dom(g)$} \\
                        \text{undefined} & \text{otherwise}
                     \end{array}\right.
	&
	\begin{array}{rcl}
		\Psi &\proves& (s^{p_1}, k, \mathsf{u}) \chequiv (s^{p_2}, j, \mathsf{u}) \textrm{ iff } \Psi(k) = \Psi(j)\\
		\Psi &\proves& (\spos{s}, k) \outcon (\spos{s}, i) \textrm{ iff } \Psi(k) = i\\
		\Psi &\proves& (\spos{s}, i) \incon (\sneg{s}, k) \textrm{ iff } \Psi(k) = i \\
		\Psi &\proves& \true \qquad \Psi\proves a \chequiv a \in \mathcal{N}
	\end{array}
\end{array}
\]
It can be easily checked that the definition is indeed a broadcast $\psi$-calculus.
We write $\Sigma_{i\in I}P$ as a shorthand for $\case\ \seq{\true}: \seq{P}$, and $P + Q$ for $\case\ \true: P \casesep \true : Q$


%




The translation is parameterised by $\rho$, which tracks the
enumeration of session prefixes, represented by multisets of asserted
names $\pass{k}$. 
The replication in $\binp{\spos{s}}{x, u^i} P$ implements the iterative broadcast receive.
We annotated the prefixes $\binp{\spos{s}}{x, u^i}^b P$ and $\rec{X^b}{P \bowtie R}$
with $b \in \set{0, 1}$ to capture their translation as a two step ($0$ and $1$)
iterative process.
The recovery process can be chosen in a non-deterministic way instead of
a $\sneg{s}$ prefix. Otherwise it is pushed in the continuation of the translation.
\[
\small
\begin{array}{l}
	\encode{\pisend{a}{\sneg{s}}P \bowtie R}_\rho =
        (\nu k)(\psiout{a}{\sneg{s}}\encode{P \bowtie R}_{\rho\cup\set{\spos{s} : k}})
	\qquad \quad
	\encode{\pirecv{a}{x} P \bowtie R}_\rho =
	(\nu k)(\psiinp{a}{x}{x}\encode{P \bowtie R}_{\rho\cup\set{\sneg{s} : k}})
	\\
	\encode{\bout{\spos{s}}{v} P \bowtie R}_{\rho \cup \set{\spos{s} : k}} =
	\psiout{(\spos{s}, k)}{v} ( \encode{P \bowtie R}_{\rho \cup \set{\spos{s} : k}} \Par \pass{k} )
	\\
	\encode{\bout{\sneg{s}}{v} P \bowtie R}_{\rho\cup\set{\sneg{s}: k }} =
	\psiout{(\sneg{s}, k, \mathsf{u})}{v}(\encode{P \bowtie R}_{\rho\cup\set{\sneg{s}:k}} \Par \pass{k} ) +
	\encode{R}_{\rho\cup\set{\sneg{s}: k }}
	\\
	\begin{array}{ll}
		\encode{\binp{\spos{s}}{x, u}^0 P \bowtie R}_{\rho\cup\set{\spos{s}: k}} =&
		\newsp{n}{\psiout{n}{u}\inact \Par !(\psiinp{n}{x}{x}
		(\psiinp{(\spos{s}, k, \mathsf{u})}{y}{y} \psiout{n}{(x\cat y)}\inact)\\
		& + \tau.(\encode{P \bowtie R}_{\rho\cup\set{\spos{s}:k}} \Par \pass{k}) )}
	\end{array}
	\\

	\begin{array}{ll}
		\encode{\binp{\spos{s}}{x, u}^1 P \bowtie R}_{\rho\cup\set{\spos{s} : k}} = &
        	\newsp{n}{(\psiinp{(\spos{s}, k, \mathsf{u})}{y}{y} \psiout{n}{(u\cat y)}\inact) +
		\tau.(\encode{P \bowtie R}_{\rho\cup\set{\spos{s}:k}}[x:=u] \Par \pass{k})   \\ 
		& \Par !(\psiinp{n}{x}{x}(\psiinp{(\spos{s}, k, \mathsf{u})}{y}{y} \psiout{n}{(x\cat y)}\inact) +
		\tau.(\encode{P \bowtie R}_{\rho\cup\set{\spos{s}:k}} \Par \pass{k}) )
		}
	\end{array}
	\\

	\encode{\binp{\sneg{s}}{x} P \bowtie R}_{\rho \cup \set{\sneg{s}:k}} =
	\psiinp{(\sneg{s}, k)}{x}{x} (\encode{P \bowtie R}_{\rho \cup \set{\sneg{s}:k}} \Par \pass{k}) + \encode{R}_{\rho \cup \set{\sneg{s}:k}}
	\\

	\encode{\bsel{\spos{s}}{l} P \bowtie R}_{\rho \cup \set{\sneg{s}:k}} =
	\psiout{(\spos{s}, l, k)}{*} (\encode{P \bowtie R}_{\rho \cup \set{\spos{s}:k}} \Par \pass{k})

        \qquad
        \encode{\pass{k}}_{\rho \cup\set{s^p : k}} = \pass{k}
	\\

	\encode{\bbra{\sneg{s}}{l_i:P_i}_{i \in I} \bowtie R}_{\rho \cup \set{\sneg{s}:k}} =
	\Sigma_{i \in I} \psiinp{(\sneg{s}, l_i, k)}{}{*} (\encode{P_i \bowtie R}_{\rho \cup \set{\sneg{s}:k}} \Par \pass{k}) + \encode{R}_{\rho \cup \set{\sneg{s}:k}}
	\\
	
	\encode{\rec{X^0}{P \bowtie R}}_{\rho} =
	\news{n} (!(\psiinp{n}{}{*} \encode{P \bowtie R}_{\rho \cup \set{\varp{X}:n}}) \Par \psiout{n}{*}\inact)
	\\
	\encode{\rec{X^1}{P \bowtie R}}_{\rho} =
	\news{n} (\encode{P \bowtie R}_{\rho \cup \set{\varp{X}:n}} \Par !(\psiinp{n}{}{*} \encode{P \bowtie R}_{\rho \cup \set{\varp{X}:n}}))
	\\

	\encode{\varp{X}}_{\rho  \cup \set{\varp{X}:n}} =
	\psiout{n}{*} \inact
	\quad
	\encode{\inact}_\rho = \inact
	\quad \ 
	\encode{\inact \bowtie R}_\rho = \inact
	\quad \ 
	\encode{P \Par Q}_\rho = \encode{P}_\rho \Par \encode{Q}_\rho
	\quad \ 
	\encode{\news{n} P}_\rho = (\nu n)\encode{P}_\rho
\end{array}
\]

The encoding respects the following desirable properties.

\begin{lemma}[Encoding Properties]\rm
	Let $P$ be a session broadcast process.

	\begin{tabular}{l}
		1. $\encode{P[x:=v]} = \encode{P}[x:=v]$ \\
		2. $\encode{P} \rightarrow Q$ implies that for a session broadcast process $P', Q \scong_\Psi \encode{P'}$.
	\end{tabular}
\end{lemma}

		\section{Broadcast Session Types}

	Broadcast session types syntax is identical to
	classic binary session type syntax
	(cf.~\cite{yoshida.vasconcelos:language-primitives}),
	with the exception that we do not allow session channel
	delegation. We assume the duality relation as defined in 
	\cite{yoshida.vasconcelos:language-primitives}. Note that
	we do not need to carry the session prefix enumeration
	in the session type system or semantics. Session prefix 
	enumeration is used operationaly only to avoid communication
	missmatch.

	\begin{tabular}{rcl}
		$S$ &$\bnfis$& $!U; S \bnfbar ?U; S \bnfbar \oplus\set{l_i:S_i}_{i \in I} \bnfbar \&\set{l_i:S_i}_{i \in I} \bnfbar \tinact
		\bnfbar \vart{X} \bnfbar \trec{X}{S}$\\
		$U$ &$\bnfis$& $\chtype{S} \bnfbar [U]$
	\end{tabular}

	\noi Typing judgements are:
	\begin{tabular}{c}
		$\Gamma \proves P$
	\end{tabular}
	read as $P$ is typed under environment $\Gamma$, with

	\begin{tabular}{c}
		$\Delta \bnfis \es \bnfbar \Delta \cat s^p:S \qquad \qquad
		\Gamma \bnfis \es \bnfbar \Gamma \cat a: \chtype{S} \bnfbar \Gamma \cat s^p : S \bnfbar \Gamma \cat \vart{X} : \Delta$
	\end{tabular}

	\noi	$\Delta$ environments map only
		session names to session types, while $\Gamma$ maps
		shared names to shared types, session names to session
		types and process variables to $\Delta$ mappings.

	The rules below define the broadcast session type system:
\[
\small
	\begin{array}{cl}
		\Gamma \cat n:U \proves n:U \ \trule{Name}
		\quad
		\tree{
			\Gamma \proves P \quad s^p \notin \fn{P}
		}{
			\Gamma \cat s^p: \tinact \proves P
		} \ \trule{Weak}
		\quad
		\tree {
			\s \notin \dom{\Gamma}
		}{
			\Gamma \proves \inact \ \trule{Inact}
		}
		\quad
		\tree {
			\Gamma \proves R \quad s^p \notin \dom{\Gamma}
		}{
			\Gamma \proves \inact \bowtie R
		} \ \trule{Recov}
	\end{array}
\]
\[
\small
	\begin{array}{ccc}
		\tree{
			\Gamma \proves a:\chtype{S} \quad \Gamma \proves \spos{s}:S \quad \Gamma \proves P
		}{
			\Gamma \cat \sneg{s}:\dual{S} \proves \pisend{a}{\sneg{s}} P
		} \ \trule{BInit}
                \quad
		\tree{
			\Gamma \proves a : \chtype{S} \quad \Gamma\cat x:\dual{S} \proves P
		}{
			\Gamma \proves \pirecv{a}{x}P
		} \ \trule{BAcc}
                \\[5mm]

		\tree{
			\Gamma \cat \spos{s} : S \proves P \bowtie R \quad \Gamma \proves v: \chtype{S'}
		}{
			\Gamma \cat \spos{s} : !\chtype{S'};S  \proves \bout{\spos{s}}{v} P \bowtie R
		} \ \trule{BSend}
                \quad
		\tree{
			\Gamma \cat \sneg{s} : S \proves P \bowtie R \quad \Gamma \proves v:\chtype{S'} \quad \sneg{s}\notin \dom{\Gamma}
		}{
			\Gamma \cat \sneg{s}: !\chtype{S'};S \proves \bout{\sneg{s}}{v} P \bowtie R
		} \ \trule{USend}
                \\[5mm]

		\tree{
			\Gamma \cat \spos{s}: S \cat x : \chtype{S'} \proves P \bowtie R  \quad \Gamma \proves u: [\chtype{S'}]
		}{
			\Gamma \cat \spos{s}: ?\chtype{S'};S \proves \binp{\spos{s}}{x, u}^b P \bowtie R
		} \trule{URcv}
		\quad
		\tree{
			\Gamma \cat \sneg{s}: S \cat x: \chtype{S'} \proves P \bowtie R \quad \sneg{s} \notin \dom{\Gamma}
		}{
			\Gamma \cat \sneg{s}: ?\chtype{S'};S \proves \binp{\sneg{s}}{x} P \bowtie R
		} \ \trule{BRcv}
                \\[5mm]

		\tree{
			\Gamma \cat \spos{s} : S_k \proves P \bowtie R \quad k \in I
		}{
			\Gamma \cat \spos{s} : \oplus\set{l_i:S_i}_{i \in I} \proves \bsel{\spos{s}}{l_k} P \bowtie R
		} \ \trule{Sel}
		\quad
		\tree{
			\Gamma \cat \sneg{s} : S_i \proves P_i \bowtie R \quad \sneg{s} \notin \dom{\Gamma}
		}{
			\Gamma \cat \sneg{s} : \&\set{l_i:S_i}_{i \in I} \proves \bbra{\sneg{s}}{l_i: P_i}_{i \in I} \bowtie R
		}\ \trule{Bra}
		\\[5mm]

		\tree{
			\Gamma_1 \proves P_1 \quad \Gamma_2 \proves P_2 \quad \spos{s} \notin \dom{\Gamma_1} \cap \dom{\Gamma_2}
		}{
			\Gamma_1 \cup \Ga_2 \proves P_1 \Par P_2
		} \ \trule{Par}
		\quad
		\tree{
			\Gamma \cat \spos{s}: S \cat \set{\sneg{s}: \dual{S_i}}_{i \in I} \proves P \quad S = S_i
		}{
			\Gamma \proves \news{s} P
		} \ \trule{SRes}
		\\[5mm]

		\tree{
			\Gamma \cat a:\chtype{S} \proves P
		}{
			\Gamma \proves \news{a} P
		} \ \trule{ShRes}
		\quad
		\tree{
			\Gamma \cup \Delta \cat \varp{X}: \Delta \proves P \quad s^p \notin \dom{\Gamma}
		}{
			\Gamma \cup \Delta \proves \rec{X^b} P
		} \ \trule{Rec}
		\quad
		\Gamma \cup \Delta \cat \varp{X}: \Delta \proves \varp{X} \ \trule{RVar}
	\end{array}
\]
	\noi Rule $\trule{Recov}$ types the recovery process. We expect no free session
	names in a recover process. Rules
	$\trule{BInit}, \trule{BAcc}, \trule{BSend}, \trule{Usend}, \trule{BRcv},
	\trule{BRcv}, \trule{Sel}$ and $\trule{Bra}$
	type prefixes in the standard way, i.e.\ check for
	object and the subject type match.
	Rule $\trule{URcv}$ types both binary instances of the unicast receive prefix
	with the same type.
	We require
	that the recovery process is carried and typed inductively in the structure of
	a process. A recovery process must not (re)use any session endpoints ($\trule{Recov}$).
	Also we require the $\sneg{s}$ to be the only one in $\Ga$. Multiple
	$\sneg{s}$ endpoints are collected using the $\trule{Par}$ rule. The $\trule{Par}$
	rule expects that there is no duplicate $\spos{s}$ endpoint present inside a process.
	When restricting a session name we check endpoint $\spos{s}$ and the set 
	of endpoints $\sneg{s}$ to have dual types. The rest of the rules are standard.

	\subsection{Soundness and Safety}

We use the standard notion of a context $\mathcal{C}$ on session types $S$ with
a single hole denoted as $[]$.
We write $\mathcal{C}[S]$ for filling a hole in $C$
with the type $S$.
We define the set of non-live sessions in a context as
$d(\Gamma) = \set{\sneg{s} : S \setbar \spos{s} : S' \in \Gamma \text{ and } \dual{S} = C[S'] \text{ with } C \not= []}$
and  live $l(\Gamma) = \Gamma \setminus d(\Gamma)$.
We say that $\Gamma$ is well typed iff $\forall \spos{s}:S \in l(\Gamma)$ then
$\set{\sneg{s}:\dual{S_i}}_{i \in I} \subset l(\Gamma)$ with $S = S_i$ or $S = ?U;S_i$. 

\begin{theorem}[Subject Congruence]\rm
	\label{thm:subject-congruence}
	If $\Gamma \proves P$ with $\Gamma$ well typed and $P \scong P'$ then
	$\Gamma \proves P'$.
\end{theorem}


\begin{theorem}[Subject Reduction]\rm
	\label{th:sr}
	\label{thm:subject-reduction}
	If $\Gamma \proves P$ with $\Gamma$ well typed, $dom(\rho)\subseteq dom(\Gamma)$ and $\encode{P}_\rho \rightarrow Q$,
        then there is $P'$ such that $\encode{P'}_\rho\psicong Q$,
	$\Gamma' \proves P'$ and $\Gamma'$ well typed with either
		$\Gamma' = d(\Gamma) \cup l(\Gamma')$ or
		$\Gamma' = d(\Gamma)\setminus \set{\sneg{s}: S} \cup l(\Gamma')$ or 
		$\Gamma' = d(\Gamma) \cup \set{\sneg{s} : S} \cup l(\Gamma')$.
\end{theorem}

\begin{definition}[Error Process]\rm
	Let $s$-prefix processes to have the following form:

	\begin{tabular}{l}
		1. $\bout{\spos{s}}{v} P$ \qquad 2. $\bsel{\spos{s}}{l} P$ \qquad 3. $\binp{\spos{s}}{x} P$ \qquad
		4. $\prod_{i \in I} \binp{\sneg{s}}{x} P_i \Par \prod_{j \in J} \context{C_j}{\binp{\sneg{s}}{x} P_j}$\\
		5. $\prod_{i \in I} \bout{\sneg{s}}{v_i} P_i \Par \prod_{k \in K} P_k \Par \prod_{j \in J} \context{C_j}{\binp{\sneg{s}}{x} P_j}$\\
		where $\prod_{i \in I} P_i \Par \prod_{k \in K} P_k \Par \prod_{j \in J} \context{C_j}{\binp{\sneg{s}}{x} P_j}$ forms an $s$-redex.\\
		6. $\prod_{i \in I} \bbra{\sneg{s}}{l_k: P_k}_{k \in K_i} \Par \prod_{j \in J} \context{C_j}{\bbra{\sneg{s}}{l_k: P_k}_{k \in K_j}}$
	\end{tabular}

	\noi with $\context{C_j}{}$ being a context that contains $\sneg{s}$ prefixes.

	A valid $s$-redex is a parallel composition of either $s$-prefixes 1 and 4, $s$-prefixes
	2 and 6, or $s$-prefixes 3 and 5. Every other combination of $s$-prefixes is invalid.
	An error process is a process of the form $P \scong \newsp{\vect{n}}{R \Par Q}$
	where $R$ is an invalid $s$-redex and $Q$ does not contain any other $s$-prefixes.
\end{definition}

\begin{theorem}[Type Safety]\rm
	A well typed process will never reduce into an error process.
\end{theorem}

\begin{proof}
	The proof is a direct consequence of the Subject Reduction Theorem~(\ref{th:sr})
	since error process are not well typed.
\end{proof}

	\section{Conclusion}
We have defined a system of session types for a calculus based on
unreliable broadcast communication. This is the first time that
session types have been generalised beyond reliable point-to-point
communication. We defined the operational semantics of our calculus by
translation into an instantiation of broadcast $\psi$-calculi, and
proved subject reduction and safety results. The use of the
$\psi$-calculi framework means that we can try to use its general
theory of bisimulation for future work on reasoning about
session-typed broadcasting systems. The definition of a session typing
system is also a new direction for the $\psi$-calculi framework.

\paragraph*{Acknowledgements}
Kouzapas and Gay are supported by the UK EPSRC project
  ``From Data Types to Session Types: A Basis for Concurrency and
  Distribution'' (EP/K034413/1). This research was supported by a
  Short-Term Scientific Mission grant from COST Action IC1201
  (Behavioural Types for Reliable Large-Scale Software Systems).  

	\bibliographystyle{eptcs}
	\bibliography{session}

\begin{thebibliography}{1}
\providecommand{\bibitemdeclare}[2]{}
\providecommand{\surnamestart}{}
\providecommand{\surnameend}{}
\providecommand{\urlprefix}{Available at }
\providecommand{\url}[1]{\texttt{#1}}
\providecommand{\href}[2]{\texttt{#2}}
\providecommand{\urlalt}[2]{\href{#1}{#2}}
\providecommand{\doi}[1]{doi:\urlalt{http://dx.doi.org/#1}{#1}}
\providecommand{\bibinfo}[2]{#2}

\bibitemdeclare{inproceedings}{DBLP:conf/lics/BengtsonJPV09}
\bibitem{DBLP:conf/lics/BengtsonJPV09}
\bibinfo{author}{Jesper \surnamestart Bengtson\surnameend},
  \bibinfo{author}{Magnus \surnamestart Johansson\surnameend},
  \bibinfo{author}{Joachim \surnamestart Parrow\surnameend} \&
  \bibinfo{author}{Bj{\"o}rn \surnamestart Victor\surnameend}
  (\bibinfo{year}{2009}): \emph{\bibinfo{title}{Psi-calculi: Mobile Processes,
  Nominal Data, and Logic}}.
\newblock In: {\sl \bibinfo{booktitle}{LICS}}, pp. \bibinfo{pages}{39--48},
  \doi{10.1109/LICS.2009.20}.

\bibitemdeclare{inproceedings}{conf/sefm/BorgstromHJRVPP11}
\bibitem{conf/sefm/BorgstromHJRVPP11}
\bibinfo{author}{Johannes \surnamestart Borgstr\"{o}m\surnameend},
  \bibinfo{author}{Shuqin \surnamestart Huang\surnameend},
  \bibinfo{author}{Magnus \surnamestart Johansson\surnameend},
  \bibinfo{author}{Palle \surnamestart Raabjerg\surnameend},
  \bibinfo{author}{Bj\"{o}rn \surnamestart Victor\surnameend},
  \bibinfo{author}{Johannes \surnamestart {{\AA}man Pohjola}\surnameend} \&
  \bibinfo{author}{Joachim \surnamestart Parrow\surnameend}
  (\bibinfo{year}{2011}): \emph{\bibinfo{title}{Broadcast Psi-calculi with an
  Application to Wireless Protocols.}}
\newblock In \bibinfo{editor}{Gilles \surnamestart Barthe\surnameend},
  \bibinfo{editor}{Alberto \surnamestart Pardo\surnameend} \&
  \bibinfo{editor}{Gerardo \surnamestart Schneider\surnameend}, editors: {\sl
  \bibinfo{booktitle}{SEFM}}, {\sl \bibinfo{series}{Lecture Notes in Computer
  Science}} \bibinfo{volume}{7041}, \bibinfo{publisher}{Springer}, pp.
  \bibinfo{pages}{74--89}, \doi{10.1007/978-3-642-24690-6\_7}.

\bibitemdeclare{article}{CGY2014}
\bibitem{CGY2014}
\bibinfo{author}{Sara \surnamestart Capecchi\surnameend},
  \bibinfo{author}{Elena \surnamestart Giachino\surnameend} \&
  \bibinfo{author}{Nobuko \surnamestart Yoshida\surnameend}
  (\bibinfo{year}{2014}): \emph{\bibinfo{title}{Global Escape in Multiparty
  Sessions}}.
\newblock {\sl \bibinfo{journal}{Mathematical Structures in Computer Science}}.
\newblock \bibinfo{note}{To appear}.

\bibitemdeclare{inproceedings}{CarboneHY08}
\bibitem{CarboneHY08}
\bibinfo{author}{Marco \surnamestart Carbone\surnameend},
  \bibinfo{author}{Kohei \surnamestart Honda\surnameend} \&
  \bibinfo{author}{Nobuko \surnamestart Yoshida\surnameend}
  (\bibinfo{year}{2008}): \emph{\bibinfo{title}{Structured Interactional
  Exceptions in Session Types}}.
\newblock In: {\sl \bibinfo{booktitle}{CONCUR}}, {\sl \bibinfo{series}{LNCS}}
  \bibinfo{volume}{5201}, \bibinfo{publisher}{Springer}, pp.
  \bibinfo{pages}{402--417}, \doi{10.1007/978-3-540-85361-9\_32}.

\bibitemdeclare{inproceedings}{honda.vasconcelos.kubo:language-primitives}
\bibitem{honda.vasconcelos.kubo:language-primitives}
\bibinfo{author}{Kohei \surnamestart Honda\surnameend},
  \bibinfo{author}{Vasco~T. \surnamestart Vasconcelos\surnameend} \&
  \bibinfo{author}{Makoto \surnamestart Kubo\surnameend}
  (\bibinfo{year}{1998}): \emph{\bibinfo{title}{Language Primitives and Type
  Disciplines for Structured Communication-based Programming}}.
\newblock In: {\sl \bibinfo{booktitle}{ESOP'98}}, {\sl \bibinfo{series}{LNCS}}
  \bibinfo{volume}{1381}, \bibinfo{publisher}{Springer}, pp.
  \bibinfo{pages}{22--138}, \doi{10.1007/BFb0053567}.

\bibitemdeclare{inproceedings}{HYC08}
\bibitem{HYC08}
\bibinfo{author}{Kohei \surnamestart Honda\surnameend}, \bibinfo{author}{Nobuko
  \surnamestart Yoshida\surnameend} \& \bibinfo{author}{Marco \surnamestart
  Carbone\surnameend} (\bibinfo{year}{2008}): \emph{\bibinfo{title}{{Multiparty
  Asynchronous Session Types}}}.
\newblock In: {\sl \bibinfo{booktitle}{POPL'08}}, \bibinfo{publisher}{ACM}, pp.
  \bibinfo{pages}{273--284}, \doi{10.1145/1328897.1328472}.

\bibitemdeclare{article}{yoshida.vasconcelos:language-primitives}
\bibitem{yoshida.vasconcelos:language-primitives}
\bibinfo{author}{Nobuko \surnamestart Yoshida\surnameend} \&
  \bibinfo{author}{Vasco~Thudichum \surnamestart Vasconcelos\surnameend}
  (\bibinfo{year}{2007}): \emph{\bibinfo{title}{Language Primitives and Type
  Discipline for Structured Communication-Based Programming Revisited: Two
  Systems for Higher-Order Session Communication}}.
\newblock {\sl \bibinfo{journal}{Electr. Notes Theor. Comput. Sci.}}
  \bibinfo{volume}{171}(\bibinfo{number}{4}), pp. \bibinfo{pages}{73--93},
  \doi{10.1016/j.entcs.2007.02.056}.

\end{thebibliography}
	\ifappendix
	\newpage
	\appendix
	\input{full_proofs}
	\fi

\end{document}